\documentclass[journal]{IEEEtran}
\usepackage{etex}

\usepackage{graphicx}
\usepackage{algorithm}
\usepackage[noend]{algpseudocode}
\usepackage{textcomp}

\usepackage{subcaption}
\usepackage{url}
\usepackage{graphicx}

\usepackage{epsfig}
\usepackage{url}
\usepackage{stfloats}
\usepackage{enumerate}
\usepackage{longtable}

\usepackage{latexsym}
\usepackage{verbatim}
\usepackage{amssymb}

\usepackage{xspace}

\usepackage{xcolor}

\usepackage{verbatim}
\usepackage{caption}
\usepackage{subcaption}
\usepackage{amsmath}
\usepackage{amssymb}
\usepackage{wrapfig}
\usepackage{tabularx}
\usepackage{lscape}
\usepackage{float}
\usepackage{inputenc}
\usepackage[utf8]{luainputenc}
\usepackage[T1]{fontenc}
\usepackage{pgfgantt}
\usepackage{footnote}

\usepackage{algpseudocode}
\usepackage{algorithm}
\usepackage[para,online,flushleft]{threeparttable}
\usepackage{multicol}
\usepackage{multirow}
\usepackage{tablefootnote}
\usepackage{mathtools}
\usepackage{graphicx}
\usepackage{caption}
\usepackage{subcaption}
\usepackage{amsmath}
\usepackage{tikz}
\usepackage{float}
\usepackage{bm}
\usepackage{url}

\newcommand{\ignore}[1]{}
\usetikzlibrary{matrix,shapes,arrows,positioning,chains, calc}

\newcommand{\specialcell}[2][c]{%
	\begin{tabular}[#1]{@{}c@{}}#2\end{tabular}}

\mathchardef\mhyphen="2D

\usepackage{pbox}

\newcommand{\algrule}[1][.2pt]{\par\vskip.5\baselineskip\hrule height #1\par\vskip.5\baselineskip}

\newcommand{\A}{$\mathcal{A}$~}

\newcommand{\Rn}{\ensuremath \stackrel{\$}{\leftarrow}\mathbb{Z}_{N}^{*}{\xspace}}

\mathchardef\mhyphen="2D
\newcommand{\sk}{\ensuremath {\mathit{sk}}{\xspace}}

\newcommand{\as}{\ensuremath {\leftarrow}{\xspace}}

\newcommand{\E}{\ensuremath {\mathcal{E}}{\xspace}}

\newcommand{\hsim}{\ensuremath {\mathit{H}\mhyphen\mathit{Sim}}{\xspace}}

\newcommand{\ro}{\ensuremath {\mathit{RO}(\cdot)}{\xspace}}

\newcommand{\Adv}{\ensuremath { \mathcal{A}}}

{\xspace}

\newcommand{\bpvsch}{\ensuremath {\texttt{BPV}}\mhyphen{\texttt{SC}}\mhyphen{\texttt{FourQ}}\mhyphen{\texttt{Schnorr}}{\xspace}}
\newcommand{\bpvschkg}{\ensuremath {\texttt{BPV}}\mhyphen{\texttt{SC}}\mhyphen{\texttt{FourQ}}\mhyphen{\texttt{Schnorr.Kg}}{\xspace}}
\newcommand{\bpvschsig}{\ensuremath{\texttt{BPV}}\mhyphen{\texttt{SC}}\mhyphen{\texttt{FourQ}}\mhyphen{\texttt{Schnorr.Sig}}{\xspace}}
\newcommand{\bpvschver}{\ensuremath{\texttt{BPV}}\mhyphen{\texttt{SC}}\mhyphen{\texttt{FourQ}}\mhyphen{\texttt{Schnorr.Ver}}{\xspace}}

\newcommand{\bpvies}{\ensuremath{\texttt{DBPV}}\mhyphen{\texttt{SC}}\mhyphen{\texttt{FourQ}}\mhyphen{\texttt{ECIES}}{\xspace}}
\newcommand{\bpvieskg}{\ensuremath{\texttt{DBPV}}\mhyphen{\texttt{SC}}\mhyphen{\texttt{FourQ}}\mhyphen{\texttt{ECIES.Kg}}{\xspace}}
\newcommand{\bpviessig}{\ensuremath{\texttt{DBPV}}\mhyphen{\texttt{SC}}\mhyphen{\texttt{FourQ}}\mhyphen{\texttt{ECIES.Enc}}{\xspace}}
\newcommand{\bpviesver}{\ensuremath{\texttt{DBPV}}\mhyphen{\texttt{SC}}\mhyphen{\texttt{FourQ}}\mhyphen{\texttt{ECIES.Dec}}{\xspace}}

\newcommand{\params}{\ensuremath {\texttt{params}}{\xspace}}

\newcommand{\BPV}{\ensuremath {\texttt{BPV}}{\xspace}}
\newcommand{\BPVOff}{\ensuremath {\texttt{BPV.Offline}}{\xspace}}
\newcommand{\BPVOn}{\ensuremath {\texttt{BPV.Online}}{\xspace}}

\newcommand{\eat}[1]{}                

\newcounter{linecounter}

\usepackage[hang,flushmargin]{footmisc}

\usepackage{pifont}
\newcommand{\C}{$\mathcal{C}$}
\newcommand{\keyRequest}{\ensuremath {\texttt{KeyQry}(\cdot)}{\xspace}}

\newcommand{\keyVer}{\ensuremath {\texttt{KeyVer}}{\xspace}}

\usepackage{wasysym}
\usepackage{amsthm}
\usepackage{hyperref}
\hypersetup{breaklinks,colorlinks,citecolor=blue,linkcolor=blue}
\def\Snospace~{\S{}}

\usepackage{tabularx,colortbl}
\newtheorem{theorem}{Theorem}
\newtheorem{def1}{Definition}

\usepackage{soul}
\newcommand{\Dronecrypt}{IoD-Crypt{\xspace}}
\begin{document}

\title{IoD-Crypt: A Lightweight Cryptographic Framework  for Internet of Drones}

\author{Muslum~Ozgur~Ozmen,~Rouzbeh~Behnia,~Attila~A.~Yavuz,~\IEEEmembership{Member,~IEEE}
\IEEEcompsocitemizethanks{\IEEEcompsocthanksitem   Muslum Ozgur Ozmen, Rouzbeh Behnia and Attila A. Yavuz are with the  Department of Computer Science and Engineering, University of South Florida, Tampa, FL,  USA. \protect\\
	E-mail: \{ozmen,behnia\}@mail.usf.edu, attilaayavuz@usf.edu
	\IEEEcompsocthanksitem Part of this work is completed when Muslum Ozgur Ozmen and Rouzbeh Behnia were with the Department of Electrical Engineering and Computer Science, Oregon State University, Corvallis, OR, USA.}
\thanks{}
}

\maketitle

\begin{abstract}

Internet of Drones (IoD) is expected to play a central role in many civilian and military applications, that require sensitive and mission-critical information to be processed. It is therefore vital to ensure the security and privacy of IoD. However, unlike traditional networks, IoD has a broader  attack surface and is highly energy-constrained, which hinder the direct adoption of standard cryptographic protocols for IoD. 

We propose an energy-efficient cryptographic framework (namely \Dronecrypt), which can potentially meet the requirements of battery-limited IoD. Specifically, \Dronecrypt~utilizes  special precomputation techniques and self-certified primitives to gain significant computation and communication efficiency over the standard public key cryptography (PKC)  suites. Our integrations and optimizations are broadly applicable to key exchange, digital signature and public key encryption schemes that  encompass  generic applications of PKC in IoD. We prove that \Dronecrypt~is secure in the random oracle model. We fully implemented \Dronecrypt~on two common drone processors, namely 8-bit AVR and 32-bit ARM, and conducted an in-depth energy analysis. Our experiments (on both platforms) showed that \Dronecrypt~offers up to 48$\times$ less energy consumption compared to standard techniques.  We have open-sourced our implementations~for wide adoption and public testing purposes.

\end{abstract}

\begin{IEEEkeywords}
	internet of drones, drone security, network security, public key cryptography, lightweight cryptography
\end{IEEEkeywords}

\IEEEpeerreviewmaketitle

\section{Introduction}\label{sec:Intro}

Internet of Drones (IoD), as an emerging mobile Internet of Things (IoT) system, is a layered network control architecture that is designed to control and coordinate unmanned aerial vehicles (UAVs or drones)~\cite{InternetOfDrones}. IoD has many applications including, but not limited to, military operations, package delivery, traffic control, environmental monitoring and disaster recovery~\cite{InternetOfDrones,Dronecrypt,IoD_SP,DroneUseCases1,Drone:IoTJournal:Planning,Drone:TMC:SecurePosition}. Due to the sensitive and strategic information involved in IoD, it is essential to guarantee their security and privacy. However,  due to the low computational power, the energy  and bandwidth limitations, this becomes  a highly challenging task~\cite{IoD_SP}. 

While there have been extensive studies on the optimization of security mechanisms on IoT, their direct adoption for IoD might not yield feasible and secure solutions. This has been further emphasized in a recent study that investigated the attacks (and their potential countermeasures) on IoT applications  \cite{IoT:Survey:2018}. In \cite{IoT:Survey:2018},  the authors point out that {\em the majority of the attack vectors could be mitigated if the existing security mechanisms and standards were properly implemented considering the requirements and characteristics of such systems.} Therefore, there is a critical need to harness the state-of-the-art security mechanisms and algorithmic optimizations to offer a viable cryptographic solution for IoD, with minimal impact to the battery lives of resource-constrained drones.

We outline the  security vulnerabilities and the  state-of-the-art solutions that are considered for aerial drones    as follows.

\subsection{Overview of Existing Cryptographic Approaches}\label{subsec:Related}

$\bullet$~\ul{\textit{Security Vulnerabilities and Solutions}:} In~\cite{Drone:ARDroneHack}, Pleban et al.  showed a successful hack and hijack  of an \emph{AR.Drone 2.0} (equipped with a 32-bit ARM processor) resulted from the lack of cryptographically secure  communication channels. Son et al. \cite{DroneSecurity::Yongdae} pointed out that most of the commodity (civilian) drone telemetry systems do not use cryptography to secure the communication. Thus, they proposed a fingerprinting method that may provide some authentication for drones.

There has been a lot of work on  the identification of the  security vulnerabilities of the current drone configurations~\cite{DroneSecurity::Yongdae2,DroneSecurity::Yongdae3,Drone:AVR1}. For instance, Shin et al. \cite{DroneSecurity::Yongdae2} showed how to extract frequency hopping sequence of FHSS-type drone controllers using a software defined radio. Son et al.~\cite{DroneSecurity::Yongdae3} demonstrated attacks that can cause drones to crash by exploiting the resonance frequency of MEMS gyroscopes.  Habibi et al.~\cite{Drone:AVR1} demonstrated  stealth attacks (with potential defense strategies combining software and hardware) that alter the sensors'  data  and even change the direction of the drone. Another vulnerability outlined in~\cite{Drone:Mavlink:Vulnerability} targets the MAVLink protocol and results in disabling the ongoing mission of the attacked drone. Garg et al. \cite{Drone:Kumar:CyberDetectionUAV} proposed a framework based on probabilistic data structures   that considers UAVs as intermediate aerial nodes to facilitate real-time analysis and cyber-threat detection mechanisms.

Recently, Lin et al.~\cite{IoD_SP} identified the differences of drone networks from traditional ones and the potential challenges towards securing drones. {\em Their study showed that the design and deployment of specifically tailored lightweight cryptographic protocols are essential for energy-constrained drones.}

$\bullet$~\ul{\textit{ Cryptographic Techniques}:} 
 The earlier studies \cite{Drones:RSA/AES:FPGA,Drones:Crypto:Weight} suggest the adoption of conventional  cryptographic techniques such as RSA and AES on   FPGAs to provide efficient solutions. Symmetric ciphers with white-box cryptography \cite{Drones:BertinoSeo:2016:Framework} were also considered in drone settings  to target the mitigations of  drone capturing attacks. Recently, certificateless cryptographic protocols were proposed for drone-based smart city applications~\cite{Drone:Elisa:Won:2017}. These protocols aim to minimize the certification overhead (e.g., transmission) that might be costly for drones. However, they might still be computationally expensive for energy-constrained drones since they include multiple elliptic curve (EC) scalar multiplications~\cite{Drone:Elisa:Won:2017}. Cheon et al.~\cite{Drone:Crypto:Homomorphic} proposed a linearly homomorphic authenticated encryption scheme to secure drone systems, but this scheme might also be too costly for resource-constrained drones, due to heavy homomorphic computations.

		There is a significant need for an {\em open-source, energy-efficient, and comprehensive cryptographic framework} which can vastly enhance the state-of-the-art crypto schemes.

\subsection{Our Contributions}\label{subsec:Contributions}

\begin{figure*}[!t]
	\centering
	\includegraphics[width=\linewidth]{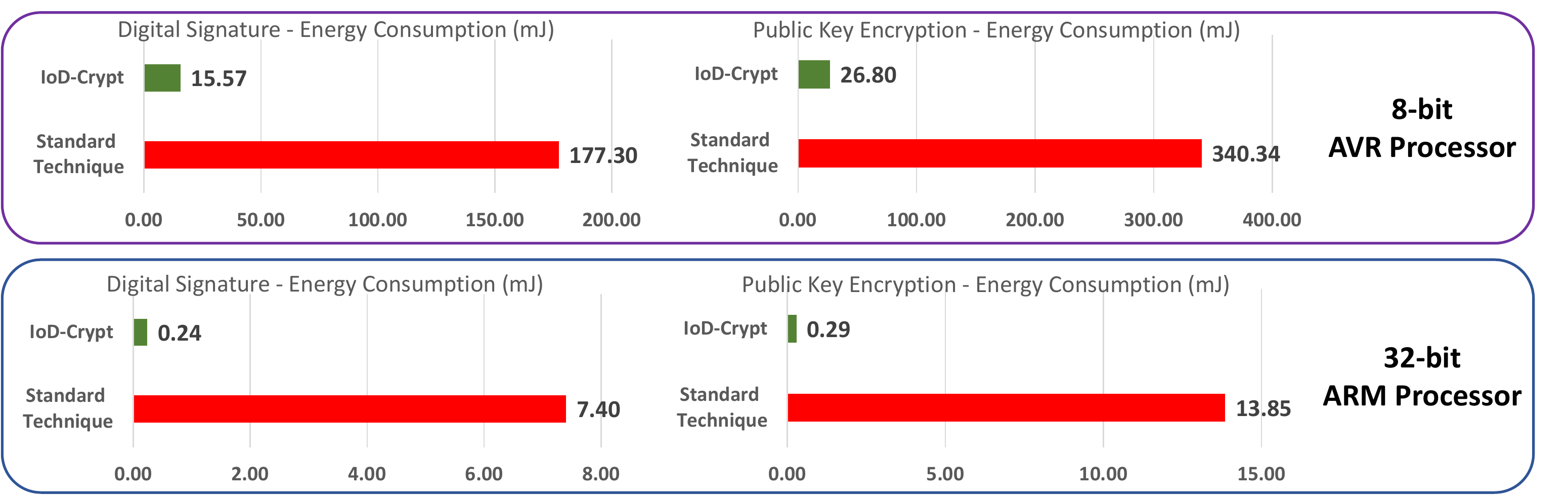}
	\caption{Energy consumption of \Dronecrypt~compared to standard PKI techniques (ECDSA, ECIES on standardized elliptic curves) on two common drone processors.}
	\vspace{-2mm}
	\label{fig:intro}
\end{figure*}

\noindent\fbox{\parbox{0.97\linewidth}{
	In this paper, we created a cryptographic framework that we refer to as {\em \Dronecrypt}, which can meet some of the stringent energy and bandwidth needs of IoD applications.
}} \\ 

Below, we give some desirable features of \Dronecrypt. \\ \vspace{-3mm} %

$\bullet$ \ul{\textit{A Comprehensive Energy-Efficient  Cryptographic Framework:}} \Dronecrypt~exploits the synergies among various cryptographic primitives such as  efficient elliptic curves and constant storage precomputation techniques while avoiding the costly certification overheads by employing a  self-certified cryptographic system. We apply our  optimizations to a wide-variety of standard suites such as a  key exchange, a digital signature, and a public key encryption. Our theoretical and experimental analysis, inlined with the requirements of IoD outlined in \cite{InternetOfDrones, IoD_SP},  confirms that \Dronecrypt~can meet the needs of IoD, specifically, in drone-to-drone and drone-to-infrastructure scenarios. \\ \vspace{-3mm}

$\bullet$ \ul{\textit{Improved Side-Channel Resiliency:}} Some of the optimization techniques we employed in our framework aim to speed up   EC scalar multiplication by converting it to EC additions~\cite{BPV:basepaper:1998, Ozmen_IOT_SP}. For instance, the digital signature ($\bpvsch$) and public key encryption ($\bpvies$) schemes that   are proposed as a part of  the \Dronecrypt~framework, do not require any EC scalar multiplication in their signing and encryption, respectively. Therefore, they achieve improved side-channel resiliency as compared to the  base schemes by avoiding  the attacks that target EC scalar multiplication \cite{SideChannelEC,SideChannel_25519}. \\ \vspace{-3mm}

$\bullet$ \ul{\textit{Cryptographic Energy Usage Analysis on Two Common Platforms:}} We identify two microprocessors, namely, 8-bit AVR~\cite{Drone:AVR1,Drone:AVR2,Drone:AVR3} and 32-bit ARM~\cite{bitcraze2016crazyflie,CrazySwarm:Drone}, from different families that are commonly used in small aerial drones (e.g., Crazyflie~\cite{bitcraze2016crazyflie}). We implemented our framework on both of these processors and assessed the energy consumption. Moreover, we tested the standard techniques and provided an in-depth comparison in \autoref{sec:Performance}. Our results showed that \Dronecrypt~offers up to $13$$\times$ and $48$$\times$ lower energy consumption on 8-bit AVR and 32-bit ARM, respectively. \\ \vspace{-3mm}

$\bullet$ \ul{\textit{Broad Applicability to Various Cryptographic Constructs:}} We identify various other (Elliptic Curve) Discrete Logarithm Problem (DLP) based schemes that can  benefit from our proposed methods  (see \autoref{subsec:ECDLP_impact}). We discuss how \Dronecrypt~can  enhance the performance of such schemes. \\ \vspace{-3mm}

$\bullet$ \ul{\textit{Open-Source Cryptographic Framework for Wide Adoption:}} We open-sourced our optimized framework for wide adoption and testing purposes at the link below.

\vspace{-.5mm}
\begin{center}
	\fbox{\url{https://github.com/ozgurozmen/IoD-Crypt}}
\end{center} 
\vspace{-.5mm}
$\bullet$ \ul{\textit{Potential   Applications to IoTs:}}  The techniques used in \Dronecrypt~can be useful for other IoT applications that require energy efficiency. Specifically, any IoT application relying on energy-limited devices (e.g., medical implantables, wireless sensors in smart cities/homes) can  benefit from \Dronecrypt.  Moreover, 8-bit AVR and 32-bit ARM processors, on which we tested \Dronecrypt, are not only common in drones but also in various other  IoT applications~\cite{Drone:ECC_IoT:Huang,IoT:Journal:ARM_AVR}. \\ \vspace{-3mm}

\noindent \textbf{Differences between this article and its preliminary versions in~\cite{Dronecrypt,Ozmen_IOT_SP}:} In this article, we harness the proposed techniques in two preliminary works to propose an integrated framework, \Dronecrypt. We provide improved implementations on drone microcontrollers that includes algorithmic optimizations. (i) We improve our framework, Dronecrypt, proposed in~\cite{Dronecrypt} to further enhance its efficiency and energy consumption, with the techniques in~\cite{Ozmen_IOT_SP}. (ii) In Dronecrypt, only 32-bit ARM processor was considered and in~\cite{Ozmen_IOT_SP}, only 8-bit AVR processor was considered. Whereas, in this version, we implement our improved framework on both 8-bit AVR drone processor~\cite{Drone:AVR1,Drone:AVR2,Drone:AVR3}, and 32-bit ARM processor (that Crazyflie 2.0 is equipped with~\cite{bitcraze2016crazyflie,CrazySwarm:Drone}). We evaluate our improved framework in terms of efficiency, energy consumption, and storage overhead. (iii) We provide the security proof of self-certified keys in~\cite{SelfCertifiedGroupECDH} for the first time in this paper, that were proposed/adopted without a security proof in~\cite{Ozmen_IOT_SP,SelfCertifiedGroupECDH}. (iv) We provide a formal security analysis of \Dronecrypt by capturing the security properties of its underlying cryptographic primitives. (v) We propose new parameter sets for our optimizations that offers higher security. (vi) We discuss  how \Dronecrypt~meets the well-established requirements of IoD networks presented in \cite{InternetOfDrones,IoD_SP} and elaborate on the broader impacts of our optimizations and proposed framework, specifically, we highlight some other EC-based schemes that can benefit from our optimization. 


\section{Building Blocks}\label{sec:Prelim}

We first outline our notation in Table~\ref{tab:Notation}.  
\begin{table}[h]
	\centering
	\caption{Notation followed to describe schemes.} \label{tab:Notation}
	\begin{threeparttable}
		\begin{tabular}{| c | c |}
			\hline
			$F_{p}$ & Finite Field \\ \hline
			$\mathbf{G}$ & Generator Group Point \\ \hline
			$N$ & A large prime \\ \hline
			$\langle x,\mathbf{U}\rangle$ & Private/Public key pair \\ \hline
			$\langle d,\mathbf{D}\rangle$ & System Wide Private/Public key pair    \\ \hline
			$\Gamma$ & Precomputation Table \\ \hline
			$\mathcal{E}_{k}$ & IND-CPA Encryption via key $k$ \\ \hline
			$\mathcal{D}_{k}$ & IND-CPA Decryption via key $k$ \\ \hline
			$\times$ & Elliptic Curve Scalar Multiplication (Emul) \\ \hline
			$\cdot$ & Multiplication in $F_q$ \\ \hline
			$\texttt{KDF}$ & Key Derivation Function \\ \hline
			$H$ & Cryptographic Hash Function $H: \{0,1\}^* \rightarrow \mathbb{Z}_N $ \\ \hline
		\end{tabular}
 EC points are shown in capital  bold letters  (e.g., $ \mathbf{B} $)
	\end{threeparttable}
\end{table}

\subsubsection{FourQ Curve~\cite{FourQBase}} FourQ is a special EC that is defined by the complete twisted Edwards equation~\cite{Bernstein::edwards_equation} $ \E/F_{p^2} : -x^2+y^2 = 1 +dx^2y^2 $. FourQ is known to be one of the fastest elliptic curves that admits $128$-bit security level~\cite{FourQBase}. Moreover, with extended twisted Edwards coordinates, FourQ offers the fastest EC addition algorithms~\cite{FourQBase}, that is extensively used in our optimizations. All of our schemes are realized on FourQ.

\subsubsection{Boyko-Peinado-Venkatesan (BPV) Generator~\cite{BPV:basepaper:1998}} BPV generator is a precomputation technique that converts   an EC scalar multiplication to multiple EC additions with the cost of a small constant-size table storage. The security of BPV on EC discrete logarithm problem  (DLP) setting is well-investigated~\cite{BPV:Ateniese:Journal:ACMTransEmbeddedSys:2017,BPV:ECC:Coron2001} and depends on the affine hidden subset sum problem. The offline (key generation) and online algorithms of BPV are described in Algorithm~\autoref{alg:BPV}.

\begin{algorithm}[h!]
	\caption{BPV Generator~\cite{BPV:basepaper:1998}}\label{alg:BPV}
	\hspace{5pt}
	\begin{algorithmic}[1]
		\Statex $\underline{(\Gamma,v,k) \as \BPVOff(1^{\kappa})}$: 
		\vspace{3pt}
		\State Generate \BPV~parameters $(v,k)$, where $k$ and $v$ are the number of pairs to be precomputed and the number of elements to be randomly selected out of $k$ pairs, respectively, for $2<v<k$.
		\State $r'_i\Rn,$~~ $\mathbf{R'}_i \as {r'_i} \times \mathbf{G}$, $i=0,\ldots, k-1$.
		\State Set precomputation table $\Gamma=\{r'_i,\mathbf{R'}_i\}_{i=0}^{k-1}$.
	\end{algorithmic}
	\algrule
	
	\begin{algorithmic}[1]
		\Statex $\underline{(r,\mathbf{R}) \as \BPVOn(\Gamma,v,k,\params)}$:
		\vspace{3pt}
		\State Generate a random set $S \subset [0,k-1]$, where $|S| = v$.
		\State $r \as \sum_{i \in S}^{} r'_i \bmod N$, $\mathbf{R} \as \sum_{i \in S}^{} \mathbf{R'}_i$.
	\end{algorithmic}
	
\end{algorithm}

\subsubsection{Arazi-Qi (AQ) Self-Certified Key Exchange~\cite{SelfCertifiedGroupECDH}} AQ self-certified key exchange minimizes the certificate transmission and verification overhead. In Algorithm~\autoref{alg:AQ}, we give key generation and shared secret algorithms where the keys are generated and distributed by a key generation center (KGC). 

\begin{algorithm}[h!]
	\caption{AQ Self-Certified (SC) Scheme~\cite{SelfCertifiedGroupECDH} }\label{alg:AQ}
	\hspace{5pt}
	\begin{algorithmic}[1]
		\Statex \underline{$(d,\mathbf{D}) \leftarrow \texttt{AQ.Setup}(1^{\kappa})$} (offline)  
		\vspace{3pt}
		\State $d \Rn$. $\mathbf{D} \xleftarrow{} d\times\mathbf{G}$
	\end{algorithmic}
	\algrule
	
	\begin{algorithmic}[1]
		\Statex \underline{$(x_a,\mathbf{U}_a) \leftarrow \texttt{AQ.Kg}(1^{\kappa}, ID_a)$} (offline)  
		\vspace{3pt}
		\State $b_a \Rn$, $\mathbf{U}_a \xleftarrow{} b_a\times\mathbf{G}$.
		\State $x_a \leftarrow [H(ID_a, \mathbf{U}_a)\cdot b_a + d]$ 
	\end{algorithmic}
	\algrule
	\begin{algorithmic}[1]
		\Statex \underline{$(\mathbf{K}_{ab}, \mathbf{K}_{ba}) \leftarrow \texttt{AQ.SharedSecret}(ID_a, \mathbf{U}_a,ID_b, \mathbf{U}_b)$ }
		\vspace{-6mm}
		
		\begin{figure}[H]
			\centering
			\begin{tikzpicture}
			\matrix (m)[matrix of nodes, column  sep=.5cm,row  sep=6mm, nodes={draw=none, anchor=center,text depth=0pt} ]{
				Node A & & Node B\\[-5mm]
				& Send $(ID_a, \mathbf{U}_a)$  & \\[-5mm]
				& &   \\[-7mm]
				& Send $(ID_b, \mathbf{U}_b)$  & \\[-5mm]
				& &  \\
			};
			
			\draw[shorten <=-.5cm,shorten >=-.5cm] (m-1-1.south east)--(m-1-1.south west);
			\draw[shorten <=-.5cm,shorten >=-.5cm] (m-1-3.south east)--(m-1-3.south west);
			\draw[shorten <=0cm,shorten >=0cm,-latex] (m-2-2.south west)--(m-2-2.south east);
			\draw[shorten <=0cm,shorten >=0cm,-latex] (m-4-2.south east)--(m-4-2.south west);
			\end{tikzpicture}
			
		\end{figure}
		\vspace{-4mm}
		\noindent Node A: $\mathbf{K}_{ab} = x_a\times [H(ID_b||\mathbf{U}_b)\times \mathbf{U}_b  + \mathbf{D}] $
		
		\noindent Node B: $\mathbf{K}_{ba} = x_b\times [H(ID_a||\mathbf{U}_a)\times \mathbf{U}_a  + \mathbf{D}] $
	\end{algorithmic}
\end{algorithm}

Nodes store the value $x_a \times \mathbf{D}$ so that the shared secret is calculated with a single EC scalar multiplication. One can also apply a technique presented in~\cite{SelfCertifiedGroupECDH} to ensure that, if required, nodes can  obtain AQ keys without exposing their private keys to the KGC. An ephemeral key exchange algorithm is also proposed using AQ keys, namely, \texttt{AQ-Hang}~\cite{Hang:2011:SPN:1998412.1998436}. In \texttt{AQ-Hang}, for each key exchange, a fresh ECDH key is calculated and added to the AQ shared secret to calculate ephemeral keys that are forward secure and self-certified.

\section{Models}\label{sec:model}

\subsection{System Model}\label{subsec:system_model}

Self-certified cryptography requires a trusted key generation center to generate and distribute the keys to the entities in the system. Therefore, we assume that there is a trusted authority that can compute and distribute the self-certified keys to the drones. This is a feasible assumption as Federal Aviation Administration (FAA) requires that all drones between $0.6 $  to $55$ lb should be registered with the U.S. Department of Transportation~\cite{InternetOfDrones}. For instance, as the trusted authority, U.S. Department of Transportation (or FAA) can issue the keys to the drones. 

\begin{figure}[!t]
	\centering
	\includegraphics[width=.80\linewidth]{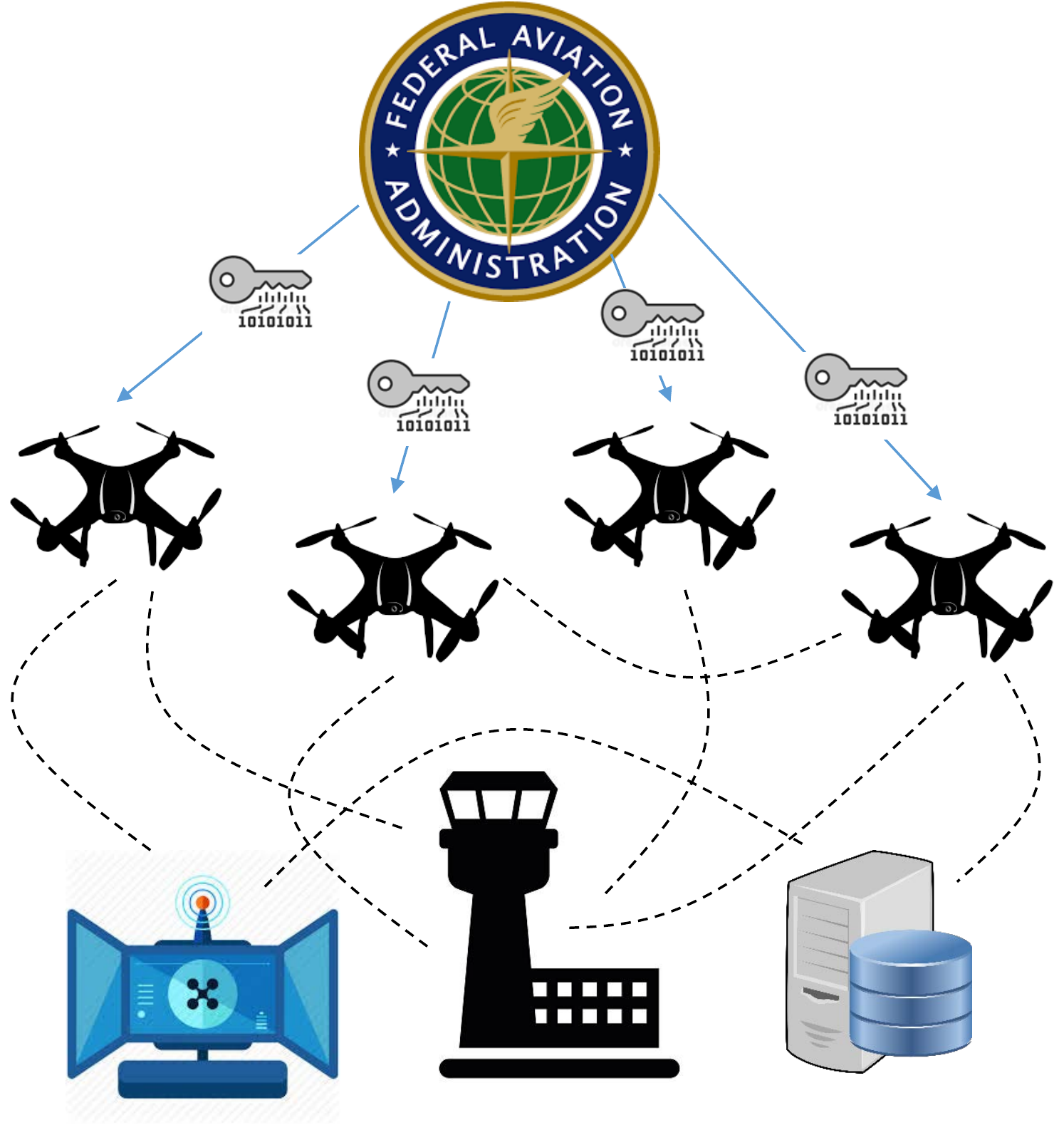}
	\caption{Our system model.}
	\label{fig:system}
\end{figure}

As depicted in \autoref{fig:system}, our main system components are the drones, and there are also recipients that communicate with the drones. These recipients may include entities such as Zone Service Providers (ZSPs), control stations or servers. For example, ZSPs provide navigation information and coordinate drones in a predetermined area, control stations can be used to send commands to non-autonomous drones, and servers can be used to store the sensors' data received from the drones~\cite{InternetOfDrones,IoD_SP,DroneUseCases1}. The communication among these entities and drones occurs over a (insecure) broadcast channel and therefore it should be secured. Note that IoD networks are energy constrained, the components are mobile and the communication range is limited~\cite{InternetOfDrones,IoD_SP}. 

\subsection{Security Model}\label{subsec:security_model}

We assume that only the KGC is trusted.  Following  \cite{Canetti2001}, we consider a probabilistic polynomial time (PPT) adversary \A  which is able  to tap the  communications between  parties and intercept/modify messages.  \A is also able to  request the key pairs of any user  $ ID $ in the system via the $ ( \mathbf{U}_{ID}, x_{ID}) \gets \texttt{KeyQry}(ID) $~oracle.   When the key pair of a user is exposed via the \keyRequest~oracle the user will be marked as  \emph{corrupted}. Since the keys are self-certified, we assume that the adversary is able to verify the validity of the keys via the \keyVer$ (\cdot) $~oracle. 

Inspired by the work of Canetti and Krawczyk \cite{Canetti2001}, we define the security of a key generation scheme in the following definition.  To capture a real-world scenario, given the system-wide public key, we assume that \A is able to observe a polynomially-bounded number of key pairs. \A is also able to verify the validity of any key pair via the \keyVer$ (\cdot) $ oracle. At some point,  \A~outputs a key pair for the target user $ T $. \A wins the experiment if the key pair is valid and  $ T $ was never queried to \keyRequest~oracle. 
\begin{def1}\label{def:KG-sec}
	A self-certified key generation scheme  $\texttt{SC-KG} = (\texttt{Setup,Kg}) $ is secure if  \A has a negligible advantage in the following experiment    $ \mathit{Expt}^{\texttt{KG-Security}}_ {\mathcal{A}} $ with a challenger $ \mathcal{C} $.
	\begin{itemize}\normalfont
		\item [--] $ (\params)\leftarrow \texttt{SC-KG.Setup}(1^\kappa) $
		\item[--] $  ( \mathbf{U}_T, x_T)\leftarrow \Adv^{ \keyRequest,\keyVer(\cdot)} (\params) $
		\item[--]  If  $ 1 \gets  \keyVer(\mathbf{U}_T,x_T) $,  return 1. Else return 0. 
	\end{itemize}
	
	\A wins if  user $ T$  was never queried to the \keyRequest~oracle. The  advantage   $  \mathit{Adv}^{\texttt{KG-Security}}_ {\mathcal{A}}$  of \A is defined as  $ { \Pr[\mathit{Expt}^{\texttt{KG-Security}}_ {\mathcal{A}} =1]   }$.
\end{def1}

We define the notion of existential unforgeability under chosen message attack (EU-CMA)   for a signature scheme $  \texttt{SGN} = (\texttt{Kg,Sig,Ver}) $   through the following definition.

\begin{def1}\label{def:EUCMA} 
 Existentially Unforgeability under Chosen Message Attack (EU-CMA)  experiment $ \mathit{Expt}^{\textit{EU-CMA}}_{\mathtt{SGN}} $   for  a signature scheme $  \texttt{SGN} = (\texttt{Kg,Sig,Ver}) $   is defined as follows.
	\begin{itemize}\normalfont
		\item [--] $ (\sk,pk)\leftarrow \texttt{SGN.Kg}(1^\kappa) $
		\item[--] $ (m^*, \sigma^*)\leftarrow \Adv^{\texttt{SGN.Sig}(\cdot)} (pk) $

	\end{itemize}
	
	\A wins the above experiment if  $  1 \leftarrow {\texttt{SGN.Ver}(m^*, \sigma^*, pk)} $ and $ m^* $ was not queried to  $ \texttt{SGN.Sig}(\cdot) $ oracle.  The EMU-CMA advantage $ \mathit{Adv}^{\textit{EU-CMA}}_{\texttt{SGN}} $  of $  \Adv  $  is defined as   $\Pr[ \mathit{Expt}^{\textit{EU-CMA}}_{\texttt{SGN}}= 1]$
\end{def1}

We define the notion of indistinguishability under chosen message attack (IND-CPA) for  encryption scheme $   {\Sigma} = (\texttt{Kg,Enc,Dec}) $   through the following definition.

\begin{def1}\label{def:INDCPA} 
	IND-CPA  experiment $ \mathit{Expt}^{\textit{INC-CPA}}_{\Sigma} $   for  an encryption  scheme $   {\Sigma} = (\texttt{Kg,Enc,Dec}) $   between the adversary \A~and a challenger \C~is defined as follows.
	\begin{itemize}\normalfont
		\item [--] $ (\sk,pk)\leftarrow \Sigma.\texttt{Kg}(1^\kappa) $
		\item[--] $ (m_0, m_1)\leftarrow \Adv^{\Sigma.\texttt{Enc}(\cdot)} (pk) $
		\item [--] \C~flips a coin $ b\gets\{0,1\} $,  computes $ c_b \gets \Sigma.\texttt{Enc}(m_b,pk)$ and returns $ c_b  $ to $ \mathcal{A}. $
		\item [--]  $b' \leftarrow \Adv^{\Sigma.\texttt{Enc}(\cdot)} (pk) $
	\end{itemize}
	
	\A wins the above experiment if   $ b=b' $.   The IND-CPA advantage $ \mathit{Adv}^{\textit{IND-CPA}}_{\Sigma} $  of $  \Adv  $  is defined as   $    { \Pr [b=b']  \leq \frac{1}{2} + \epsilon }$. 
\end{def1}

\section{Proposed Framework}\label{sec:Proposed}

\noindent \textbf{Design  Rationale:} Our main idea is as follows.

(i) The transmission and verification of certificates introduce a significant communication and computation burden which might be expensive for many IoD applications. Moreover, drones are required to be registered with  central authorities (e.g., FAA). We observed that self-certified cryptographic keys as in~\cite{SelfCertifiedGroupECDH} can be used to bridge this gap and enhance the security and performance of IoD. Specifically, we integrate self-certified keys into Schnorr signatures~\cite{Schnorr91}  and Elliptic Curve Integrated Encryption Scheme (ECIES)~\cite{Pointcheval00psec3}  to achieve efficient authentication   and  confidentiality tools. We    refer to such improved schemes as $\bpvsch$~and $\bpvies$, respectively. Therefore, our approach eliminates the burden of certificates and  harnesses a central trusted authority for the initial key distribution that fits  IoD in compliance with the FAA requirements.

(ii) We alleviate  the computational overhead of our schemes by harnessing special precomputation techniques such as the BPV precomputation technique (see Algorithm \ref{alg:BPV}). We adopt a variation of the BPV technique, the Designated BPV (i.e., DBPV)~\cite{Ozmen_IOT_SP}, which permits an efficient use of BPV in ECIES. Finally, we realize all of our improved schemes on FourQ curve~\cite{FourQ}, which is one of the most computationally efficient ECs. Specifically, the efficiency of EC additions in FourQ complements BPV/DBPV techniques.

(iii) The AQ self-certified system \cite{SelfCertifiedGroupECDH} has not shown to be provably secure. We fill this research gap by providing a formal security proof (in the random oracle model) for the AQ self-certified keys. We then analyze the security of our proposed schemes by capturing the integration of self-certified keys and DBPV into Schnorr and ECIES. 
Our proposed framework consists of a signature scheme and a CPA secure encryption scheme to achieve authentication and confidentiality, respectively.

\subsubsection{$\bpvsch$}We adopt the BPV technique in the Schnorr signature scheme to speed up the EC scalar multiplication required in Schnorr. Moreover, we use AQ public/private key pair in our Schnorr algorithms to  avoid the burden of communicating certificates by achieving self-certification. In PKI-based schemes, before the signature verification, the public key of the signer (i.e., via its certificate) should be verified.  Note that $[H(ID_b||\mathbf{U}_b)\times \mathbf{U}_b  + \mathbf{D}]$ part in the verification step is equivalent to verifying a certificate without any extra communication cost. Moreover, it can be calculated and stored before the online steps. Our scheme is provided in Algorithm~\autoref{alg:BPV-Schnorr}.

\begin{algorithm}[h!]
	\caption{$\bpvsch$ with AQ Keys}\label{alg:BPV-Schnorr}
	\hspace{5pt}
	\begin{algorithmic}[1]
		\Statex $\underline{(\Gamma,x,\mathbf{U}) \leftarrow \bpvschkg(1^{\kappa})}$: 
		\vspace{3pt}
		\State KGC generates its keys $d \Rn$, $\mathbf{D} \xleftarrow{} d\times\mathbf{G}$
		\State $b \Rn$, $\mathbf{U} \xleftarrow{} b\times\mathbf{G}$.
		\State $x \leftarrow [H(ID, \mathbf{U})\cdot b + d]$. 
		\State Generate $(v,k)$ as in $\BPVOff(\cdot)$ Step 1
		\State $r'_i\Rn,$~~ $\mathbf{R'}_i \as {r'_i} \times \mathbf{G}$, $i=0,\ldots, k-1$.
		\State Set precomputation table $\Gamma=\{r'_i,\mathbf{R'}_i\}_{i=0}^{k-1}$.
	\end{algorithmic}
	\algrule
	
	\begin{algorithmic}[1]
		\Statex $\underline{(s,e)\leftarrow \bpvschsig(m,x,\Gamma)}$:
		\vspace{3pt}
		\State Generate a random set $S \subset [0,k-1]$, where $|S| = v$.
		\State $r \as \sum_{i \in S}^{} r'_i \bmod N$, $\mathbf{R} \as \sum_{i \in S}^{} \mathbf{R'}_i$.
		\State $e\as H(m||\mathbf{R}),~~s\as (r-e\cdot x) \bmod N$ 
	\end{algorithmic}
	\algrule
	
	\begin{algorithmic}[1]
		\Statex $\underline{b\leftarrow \bpvschver(m,\langle s,e \rangle,\mathbf{U}, ID)}$: 
		\vspace{3pt}
		\State $\mathbf{R'}\as  e \times [H(ID||\mathbf{U})\times \mathbf{U}  + \mathbf{D}] + s \times \mathbf{G}$
		\State If $e=H(m||\mathbf{R'})$ then  set $b=1$ as {\em valid}, else $b=0$ .
	\end{algorithmic}
\end{algorithm}

\subsubsection{$\bpvies$}ECIES is a public key encryption scheme also included in TinyECC~\cite{TinyECCLiu:2008}, that offers forward-secure authentication and encryption. In the standard ECIES scheme, two EC scalar multiplications are required, one to generate ephemeral key pair, and one to generate a one-time key. We also adopt the Designated BPV technique that allows to speed up both of these EC scalar multiplications in ECIES by replacing them with EC additions~\cite{Ozmen_IOT_SP}. \\ \vspace{-3mm}

The BPV precomputation technique is extended for ``Designated'' elliptic curve points, where not only the EC scalar multiplication over the generator ($\mathbf{G}$) but also over other precomputed  curve points are stored~\cite{Ozmen_IOT_SP}. This significantly benefits the  public key encryption schemes (e.g., ECIES~\cite{Pointcheval00psec3}) where an EC scalar multiplication is required over the receiver's public key. Although DBPV eliminates EC scalar multiplications, the caveat is the small storage overhead. However, public key encryption schemes are usually suitable for applications where multiple entities (e.g., drones) send messages/reports to a single command center/base station (e.g., ZSPs) using their public key. This allows easy adoption of DBPV to public key encryption schemes (especially in the  IoD networks) since the storage overhead introduced would be tolerable, which is confirmed by our experiments (see \autoref{sec:Performance}).

For the  ECIES scheme, as   in Algorithm~\autoref{alg:DBPV-ECIES}, we generate the designated precomputation table over $[H(ID||\mathbf{U})\times \mathbf{U}  + \mathbf{D}]$, that is based on the public key of the receiver.

\begin{algorithm}[h!]
	\caption{$\bpvies$ with AQ Keys}\label{alg:DBPV-ECIES}
	\hspace{5pt}
	\begin{algorithmic}[1]
		\Statex $\underline{(\Gamma,x,\mathbf{U}) \leftarrow \bpvieskg(1^{\kappa})}$:  
		\vspace{3pt}
		\State KGC generates its keys $d \Rn$, $\mathbf{D} \xleftarrow{} d\times\mathbf{G}$
		\State $b \Rn$, $\mathbf{U} \xleftarrow{} b\times\mathbf{G}$.
		\State $x \leftarrow [H(ID, \mathbf{U})\cdot b + d]$. 
		\State Generate  $(v,k)$ as in $\BPVOff(\cdot)$ Step 1
		\State $r'_i\Rn,$~ $\mathbf{R'}_i \as {r'_i} \times \mathbf{G}$,~ $\mathbf{S'}_i \as {r'_i} \times \mathbf{[H(ID||\mathbf{U})\times \mathbf{U}  + \mathbf{D}]}$, $i=0,\ldots, k-1$.
		\State Set precomputation table $\Gamma=\{r'_i,\mathbf{R'}_i, \mathbf{S'}_i\}_{i=0}^{k-1}$.
	\end{algorithmic}
	\algrule
	
	\begin{algorithmic}[1]
		\Statex $\underline{(c,d,\mathbf{R})\leftarrow \bpviessig (m,\mathbf{U}, ID, \Gamma)}$:
		\vspace{3pt}
		\State Generate a random set $S \subset [0,k-1]$, where $|S| = v$.
		\State $r \as \sum_{i \in S}^{} r'_i \bmod N$, $\mathbf{R} \as \sum_{i \in S}^{} \mathbf{R'}_i$,  $\mathbf{S} \as \sum_{i \in S}^{} \mathbf{S'}_i$.
		\State $(k_{enc}, k_{MAC}) \leftarrow \texttt{KDF}(\mathbf{S})$
		\State $c \leftarrow \mathcal{E}_{k_{enc}}(m)$
		\State $d \leftarrow \texttt{MAC}_{k_{MAC}}(c)$
	\end{algorithmic}
	\algrule
	
	\begin{algorithmic}[1]
		\Statex $\underline{m\leftarrow \bpviesver (x,\langle c,d,\mathbf{R} \rangle)}$:
		\vspace{3pt}
		\State $\mathbf{S'}\as y \times \mathbf{R}$
		\State $(k_{enc}, k_{MAC}) \leftarrow \texttt{KDF}(\mathbf{S'})$
		\State If $d \neq  \texttt{MAC}_{k_{MAC}}(c)$ return {\em invalid}
		\State $m \as \mathcal{D}_{k_{enc}}(c)$
	\end{algorithmic}
\end{algorithm}

We also complement our optimized public-key primitives by adopting lightweight symmetric ciphers. We consider Chacha stream cipher and Poly1305 authenticator~\cite{rfc7539}, due to their fast operations and well-studied security~\cite{CHACHAPOLYSecurity}.

\subsection{Broad Application to IoD and EC-based Schemes}\label{subsec:ECDLP_impact}

We discuss how \Dronecrypt~can meet some of the needs of IoD networks as outlined in~\cite{InternetOfDrones, IoD_SP}. First and foremost, the energy efficiency of \Dronecrypt~makes it a suitable choice for IoD. Our experiments showed that \Dronecrypt~offers high energy efficiency on small drones that are equipped with 8-bit and 32-bit processors (see~\autoref{sec:Performance}). Since IoD networks consist of  energy-constrained small drones~\cite{IoD_SP}, \Dronecrypt~techniques are preferable over the standard techniques. Moreover, since  the IoD networks usually use broadcasting for message delivery and may carry highly sensitive information, public key cryptographic tools may be useful due to their scalability, non-repudiation and public verifiability.

Moreover, the proposed techniques in \Dronecrypt~framework can benefit various other ECDLP based schemes, that can be useful for securing IoD networks. For instance, in~\cite{Drone:Elisa:Won:2017}, certificateless key encapsulation and encryption techniques are proposed to secure drone communications in smart cities. The EC scalar multiplications in the certificateless hybrid encryption scheme in~\cite{Drone:Elisa:Won:2017} can be accelerated with DBPV. Similarly, it can also benefit from the adoption of FourQ curve, especially when considered with the DBPV due to its optimized EC addition. In~\cite{Drone:Signature_Huang,IoT:Signcryption}, efficient digital signature and signcryption schemes were proposed. Although these schemes include bilinear pairing and therefore, cannot adopt FourQ curve, they can benefit from \Dronecrypt~optimizations. In~\cite{Drone:Signature_Huang}, signer needs to perform an EC scalar multiplication that can be accelerated with  BPV, and in~\cite{IoT:Signcryption}, at the signcryption algorithm, DBPV can be utilized to speed up a couple of EC scalar multiplications that are necessary to generate an ephemeral key.

\subsection{Limitations and Potential Remedies} \label{secx:Limitations}
High energy and computational efficiency of \Dronecrypt~comes with a trade-off due to the adopted optimization methods. First, the precomputation technique used in \Dronecrypt~requires storing a constant-size table at the signer/sender. However, we show that with the correct parameter choices, it is even possible to store this table in highly resource-constrained 8-bit AVR processors. Moreover, \Dronecrypt~requires a trusted authority (e.g., KGC) to compute and deliver keys to the drones, to remove the certificate transmission and verification overhead. As explained in our system model in~\autoref{subsec:system_model}, since drones should be registered with authorities (e.g., FAA or U.S. Department of Transportation), we believe this is a plausible assumption for the envisioned IoD applications.

\section{Security Analysis}\label{sec:Security}

\begin{theorem}\label{thm:KG-sec}
	If an adversary \A~wins the experiment in Definition \autoref{def:KG-sec}  with non-negligible probability after making  $ q_k $ key queries and $ q_h $ hash queries,  then one can build another algorithm \C~that runs \A as a subroutine and can solve the Diffie-Hellman Problem with a non-negligible probability.
	
\end{theorem}

\begin{proof}
	\C~runs the  $\texttt{AQ.Setup}(\cdot)$ algorithm and passes $ \mathbf{D} $ to \Adv.    \C~aims to solve the Discrete Logarithm (DL) problem on the input of   $ \mathbf{D} $.  \C~keeps tables $ L_h  $ and $L_k$ to keep track of  random oracle  queries and key queries,  respectively.

	\noindent \underline{  Setup \ro~Oracle}: \C~implements a function \hsim~to handle \ro~queries to random oracles $ H $. That is, the cryptographic hash functions $ H  $ are modeled as random oracles via \hsim~as follows.
	\begin{enumerate}[i)]
		
		\item $h \as \hsim(ID||\mathbf{U} _{ID},L_h)$: If $R \in L_h$ then \hsim~returns the corresponding value $h \as L_m(ID||\mathbf{U} _{ID})$. Otherwise, it returns $ h  \Rn  	 $ as the answer, and  inserts $(ID||\mathbf{U} _{ID},h )$ into $L_m$.

	\end{enumerate}
	\noindent  \underline{Queries of  $\mathcal{A}$}: \A~can query \ro~and \keyRequest~oracles on any message of its choice up to $q_H$ and $q_k$ times, respectively. 
	\begin{enumerate}[1)]
		\item {\em Handle \ro~queries}: \A's queries on $ H  $ is handled by \hsim~function as described above.
		
		\item {\em Handle \keyRequest~queries}: To answer \A's~  key queries on any user of its choice $ID$, \C~inserts $ID$ into $  {L_k} $~and  continues as follows.
		
	\end{enumerate}
	\begin{enumerate}[i)]
		\setlength{\itemsep}{2pt}
		\setlength{\parskip}{0pt}
		\setlength{\parsep}{0pt}
		
		\item Pick  $ x_{ID} \Rn  $ and $ h \Rn  $. 
		\item Compute $ \mathbf{U}_{ID} \as  h^{-1}\times (x_{ID} \times \mathbf{G}- \mathbf{D})  $.
		\item \C~checks if $ L_k (ID||\mathbf{U}_{ID}) \neq \bot$, \C~{\em aborts}  and outputs $ 0 $.  Else, it sets $L_k (ID||\mathbf{U}_{ID})   \gets h$. 
		\item Lastly,  \C~outputs $ (x_{ID}, U_{ID}) $ and sets $ L_k(ID) \gets (x_{ID}, U_{ID}) $.
		
	\end{enumerate}
	
	\noindent\underline{Output of~$\mathcal{A}$}: Finally, $\mathcal{A}$ outputs a forgery key pair for any user of its choice $ ID^* $~as $( x_{ID^*}, \mathbf{U}_{ID^*})$. By Definition \ref{def:KG-sec}, \A~wins the if the below conditions hold.
	
	\begin{enumerate}[i)]
		\item $ 1 \gets \keyVer(x_{ID^*}, \mathbf{U}_{ID^*},ID^*)$
		\item $ID^{*} \notin L_k$
	\end{enumerate}
	
	\noindent \underline{{$ \mathcal{C}$'s Solution to the DH problem}}:  If  \A~fails to output a key pair], \C~ also fails in  solving the DL problem, and therefore, \C~{\em aborts} and returns $0$. 
	
	Otherwise,  if  \A~outputs a successful forgery $( x_{ID^*}, \mathbf{U}_{ID^*})$,  using the forking lemma  \cite{Bellare-Neven:2006},  $ \mathcal{C} $ can rewind $ \mathcal{A} $ to get a second forgery for user $ ID^* $ as $ ( \tilde{x}_{ID^*} , \tilde{\mathbf{U}}_{ID^*}) $ where $ x_{ID^*} \neq \tilde{x}_{ID^*} $ and  $  \mathbf{U}_{ID^*}= \tilde{\mathbf{U}}_{ID^*}$ with an overwhelming probability. Given~the forgeries $( x_{ID^*}, \mathbf{U}_{ID^*})$ and  $ ( \tilde{x}_{ID^*} , \tilde{\mathbf{U}}_{ID^*}) $  on $ ID^* $,  based on  \cite{Bellare-Neven:2006}, we know that  $ H(ID^*||\mathbf{U}_{ID}) \neq H(ID^*||\tilde{\mathbf{U}}_{ID})  $.     
	
	Given $ 1 \gets \keyVer( {x}_{ID^*},  {\mathbf{U}}_{ID^*},ID^*)$  and $ 1 \gets \keyVer(\tilde{x}_{ID^*},$ $ \tilde{\mathbf{U}}_{ID^*},ID^*)$,  \C~computes    $ b \gets  \frac{x_{ID^*}-\tilde{x}_{ID^*}}{H(ID^*||\mathbf{U}_{ID}) + H(ID^*||\tilde{\mathbf{U}}_{ID})}$ and solve for the discrete logarithm of $\mathbf{D} $ as  $ d \gets   x_{ID^*} -  H(ID^*||\mathbf{U}_{ID})  \cdot b $.
\end{proof}

\begin{theorem}\label{thm:Schnorr}
	The signature scheme proposed in Algorithm \autoref{alg:BPV-Schnorr} is EU-CMA secure  in the sense of  Definition \autoref{def:EUCMA}. 
	
\end{theorem}
\begin{proof}
	
Based on the results of   \autoref{thm:KG-sec}, we know that the     self-certified key     used in Algorithm \autoref{alg:BPV-Schnorr} is secure.  Moreover, The distribution of   keys generated by Algorithm \autoref{alg:AQ} is indistinguishable from   uniform random due \cite{Schnorr91}. As shown in the previous works \cite{BPV:basepaper:1998,BPV:Ateniese:Journal:ACMTransEmbeddedSys:2017}, one can see that the  distribution of \BPV~output $r$ is {\em statistically close} to the uniform random  with the appropriate choices	of parameters $(v,k)$. Following this, if  the ephemeral randomness $r$ in  $ \bpvschsig $ ~Step 2 is obtained from \BPV~generator, then $ \bpvsch  $ is secure  given the hardness of the  {\em Affine Hidden Subset Sum} problem~\cite{BPV:Ateniese:Journal:ACMTransEmbeddedSys:2017,HiddenSubsetSum:BPV:Analysis:Nguyen:Crypto1999}. Therefore, based on these results, the  signature scheme proposed in Algorithm \autoref{alg:BPV-Schnorr} is secure in the sense of Definition \autoref{def:EUCMA} under the DL problem.
\end{proof}

\begin{theorem}\label{thm:ECIES}
	The encryption scheme proposed in Algorithm \autoref{alg:DBPV-ECIES} is IND-CPA secure  in the sense of  Definition \autoref{def:INDCPA}. 
	
\end{theorem}
\begin{proof}
%
	The security and properties of the self-certified keys adopted in   Algorithm \autoref{alg:DBPV-ECIES}  follows from the first part of the proof in \autoref{thm:Schnorr}. 
	The original BPV, as shown in Algorithm \autoref{alg:BPV} computes $ \mathbf{R} $ by  multiplying a random $ r $ with the generator of the group $ \mathcal{E} (\mathcal{F}_{p^2}[N])$. However, due to the property of the group $ \mathcal{E} (\mathcal{F}_{p^2}[N])$, in the DBPV method used in Algorithm \autoref{alg:DBPV-ECIES},  the public key  of the user, $ (x\cdot H(ID||\mathbf{U})\cdot b) \times \mathbf{G}$, is also a generator of the group. Therefore, proving the randomness of the BPV output (as also studied in  \cite{BPV:basepaper:1998,BPV:Ateniese:Journal:ACMTransEmbeddedSys:2017}) for the DBPV method is identical to the original BPV method. Based on these results,       the encryption  scheme proposed in Algorithm \autoref{alg:DBPV-ECIES} is secure in the sense of Definition \autoref{def:INDCPA}, and the proof can be obtained  identical as in~\cite{Pointcheval00psec3}.
\end{proof}

\noindent \textbf{Parameter Choice: }The security of our proposed schemes depend on the underlying curve and BPV/DBPV parameter choice of $(v,k)$. As the underlying curve, we select FourQ, that offers $128$-bit security level and very fast curve operations such as EC addition~\cite{FourQBase}. Security of BPV/DBPV depends on the $v$-out-of-$k$ different combinations that could be created with the precomputed table. Moreover, the selection of $v$ and $k$ have a trade-off between storage and computation, where increasing $k$ increases the storage and increasing $v$ increases the computation. We offer two parameter sets, specifically $v = 28$, $k = 256$ and $v = 18$, $k = 1024$ that offer $2^{128}$ different combinations, and therefore 128-bit security level. From these two, we focus on $v = 28, k = 256$ since it requires $4$$\times$ less storage, which can easily fit highly resource-constrained processors such as 8-bit AVR.

\section{Performance Analysis and Comparison}\label{sec:Performance}

We fully implemented \Dronecrypt~on two common drone microprocessors and assessed the performance and energy efficiency of our schemes. We measured the energy consumption with the formula $E = V \cdot I \cdot t$, following the works~\cite{Ozmen_IOT_SP,Dronecrypt,FourQ8bit}. We compare the costs of our optimized techniques with the standard ones, specifically, we implemented ECDH, ECDSA and ECIES protocols on the standardized secp256k1 curve, and adopted AES and HMAC as the symmetric ciphers.

\begin{table*}[t!]
	\centering
	\small
	\caption{Comparison of Public Key Primitives on 8-bit AVR Processor} \label{tab:AVR:PublicKey}
	\resizebox{\textwidth}{!}{
		\begin{threeparttable}
			\begin{tabular}{| c || c | c |  c | c | c | c |}
				\hline
				\textbf{Protocol} & \textbf{CPU Cycles} & \specialcell[]{\textbf{CPU Time} \\ \textbf{($s$)}} & \specialcell[]{\textbf{Memory\tnote{$ \dagger $}} \\ \textbf{($Byte$)}} & \specialcell[]{\textbf{Bandwidth} \\ \textbf{($Byte$)}} & \specialcell[]{\textbf{Energy} \\ \textbf{Consumption ($mJ$)}} & \specialcell[]{\textbf{Certificate} \\ \textbf{Overhead}} \\ \hline
				
				\multicolumn{7}{|c|}{\textbf{\em Standard Techniques}} \\ \hline 
					\rowcolor{lightgray}
				ECDH & $ 25,840,000 $ & $ 1.61  $& $ 32 $ & $ 32 $ & $ 161.52 $ & Yes \\	
				
				Ephemeral ECDH & $ 54,380,000 $ & $ 3.40 $ & $ 32 $ & $ 32 $ & $ 339.86 $ & Yes \\
					\rowcolor{lightgray}
				ECDSA-Sign & $ 28,370,000 $ &$  1.77 $ & $ 32 $ & $ 64 $ & $ 177.30 $ & Yes \\
				
				ECDSA-Verify & $ 28,960,000 $ &$  1.81 $ & $ 32 $ & $ 64 $ & $ 181.01 $ & Yes \\
					\rowcolor{lightgray}
				ECIES-Encrypt & $ 54,450,000 $ & $ 3.40 $ & $ 32 $ &  $ 32 $ + |c| + |MAC|  & $ 340.34 $ & Yes \\
				
				ECIES-Decrypt & $ 25,920,000 $ & $ 1.62 $ & $ 32 $ &   $ 32 $ + |c| + |MAC| & $ 161.99 $ & Yes \\ \hline
										
				\multicolumn{7}{|c|}{\textbf{\em \Dronecrypt}} \\ \hline 
				\rowcolor{lightgray}
				
				\texttt{AQ} &$  6,940,000 $ &$  0.43 $ & $ 32 $ & $ 32 $ &$  43.38 $ &  No \\
				
				\texttt{BPV-AQ-Hang} & $ 9,140,000 $ & $ 0.57 $ & $ 16416 $ & $ 32 $ &$  57.14 $ &  No \\
					\rowcolor{lightgray}
				$\bpvschsig$ &$  2,490,000 $ & $ 0.15  $& $ 16416 $ & $ 64 $ & $ 15.57 $ & No \\
				
				$\bpvschver$ & $ 8,310,000 $  & $ 0.52 $ & $ 32 $ & $ 64 $ & $ 51.94 $ & No\\
					\rowcolor{lightgray}
				$\bpviessig$ & $ 4,290,000 $ &$  0.27 $ & $ 24608 $ & $ 32 $ + |c| + |MAC| & $ 26.80 $  & No \\
				
				$\bpviesver$ & $ 6,980,000 $ & $ 0.44 $ & $ 32 $ & $ 32 $ + |c| + |MAC| & $ 43.61 $ & No \\ \hline

			\end{tabular}
			\begin{tablenotes}[flushleft] \scriptsize{
					$\dagger $ Memory denotes private key size for sign/encrypt schemes as signer/sender stores it, memory denotes public key size for verify/decrypt schemes as verifier/receiver stores it.
				}
			\end{tablenotes}
		\end{threeparttable}
	}
\end{table*}

\begin{table}[t!]
	\centering
	\small
	\caption{Comparison of Symmetric Ciphers on 8-bit AVR} \label{tab:AVR:Symmetric}
	\vspace{-1mm}
	\resizebox{0.49\textwidth}{!}{
	\begin{threeparttable}
		\begin{tabular}{| c || c | c | c | c |}
			\hline
			\textbf{Protocol} & \textbf{($KB/s$)} & \specialcell[]{\textbf{CPU} \\ \textbf{Cycles\tnote{$ \dagger $}}} & \specialcell[]{\textbf{CPU} \\ \textbf{Time ($ms$)}} &  \specialcell[]{\textbf{Energy} \\ \textbf{Cons. ($\mu J$)}} \\ \hline

			\multicolumn{5}{|c|}{\textbf{\em Standard Techniques}} \\ \hline 
			
			AES & $ 29.339 $ & $ 19,500 $ & $ 1.22 $ & $ 122.34 $\\
			
				\rowcolor{lightgray}
			AES-GCM & $ 8.890 $ & $ 76,500  $& $ 4.78 $ &$  478.15 $ \\
			
			HMAC\tnote{$ \ddagger $} & $ 22.017 $ & $ 182,200 $ &$  11.39 $ & $ 1138.69 $ \\ \hline
			
			\multicolumn{5}{|c|}{\textbf{\em \Dronecrypt}} \\ \hline 
			
			CHACHA20 & $ 65.670 $ & $ 8,300 $ & $ 0.52 $ & $ 521.89 $\\
			\rowcolor{lightgray}
			CHACHA-POLY & $ 23.687 $ & $ 35,500 $ & $ 2.22 $ & $ 222.03 $ \\
			
			POLY1305 & $ 37.169 $ & $ 21,400 $ & $ 1.34 $ & $ 133.81 $\\  \hline

		\end{tabular}
		\begin{tablenotes}[flushleft] \scriptsize{
				$\dagger $ CPU cycles presented here are for a $32$-byte message. \\
				$ \ddagger $ SHA256 is used as the standard hash function for HMAC.
			}
		\end{tablenotes}
	\end{threeparttable}
	}
\end{table}

\subsection{Performance on 8-bit AVR Processor}

\noindent \textbf{Hardware Configurations: }We implemented \Dronecrypt~targeting an 8-bit AVR ATmega 2560 processor. ATmega 2560 has a maximum frequency of $16$ MHz and it is equipped with $8$ $KB$ SRAM and $256$ $KB$ flash memory. ATmega 2560 operates at $V = 5$ $V$ and $I = 20$ $mA$ while running at $16$ $MHz$ (full capacity), that is used to estimate the energy consumption. Since we used IAR Embedded Workbench to develop \Dronecrypt, our codes can be easily tested/used in other 8-bit AVR processors.

\noindent \textbf{Implementation: }We implemented the public key primitives of \Dronecrypt~using the open-source implementation of FourQ library on microprocessors~\cite{FourQ8bit}, which provides the basic EC operations such as EC scalar multiplication, EC addition. As for the symmetric primitives, we used   Weatherley's   library~\cite{weatherley}. We used the cycle-accurate simulator of IAR Embedded Workbench to benchmark our schemes, as in~\cite{FourQ8bit}. For the standard public key primitives, we used micro-ecc~\cite{microECC}, that offers the lightweight implementations of standard elliptic curves for microprocessors.

The results of our experiments are presented in~\autoref{tab:AVR:PublicKey} for the public key primitives and in~\autoref{tab:AVR:Symmetric} for the symmetric ones. As shown, \Dronecrypt~primitives outperform their counterparts in computation time and energy consumption. On the other hand, they require a storage of $16$ $KB$ for the adoption of  BPV and $24$ $KB$ for DBPV. However, these precomputation tables are stored as a part of the flash memory of ATmega 2560 and they take less than $10\%$ of the memory. Moreover, the standard public key techniques take a few seconds, that might be infeasible for time-critical IoD applications, whereas \Dronecrypt~signature generation takes $150$ $ms$ and public key encryption takes $270$ $ms$.

\subsection{Performance on 32-bit ARM Processor}

\noindent \textbf{Hardware Configurations:} We targeted an actual drone for our tests on 32-bit ARM processor, namely Crazyflie 2.0~\cite{bitcraze2016crazyflie}. Crazyflie 2.0 is equipped with an STM32F4 processor (32-bit ARM Cortex M-4 architecture) that has $192$ $KB$ SRAM and $1$ $MB$ flash memory, and operates at $168$ $MHz$. At max frequency, the processor requires $3.3$ $V$ and $40$ $mA$.

\noindent \textbf{Implementation: }We used the Microsoft FourQ library to implement \Dronecrypt~for 32-bit ARM processor~\cite{FourQBase,FourQ8bit}. This library   provided the basic EC operations which we used to build the protocols with the optimization techniques. For our counterparts, we  used micro-ecc~\cite{microECC} and for symmetric primitives, we used Wolfcrypt library~\cite{WolfCrypt}.

\autoref{tab:ARM:PublicKey} and~\autoref{tab:ARM:Symmetric} show the benchmarks of \Dronecrypt~and the standard techniques on STM32F4 processor. Our experiments showed that \Dronecrypt~offers up to $48$$\times$ less energy consumption. Moreover, \Dronecrypt~can support to very high message throughputs due to its high efficiency. More specifically, with \Dronecrypt, an STM32F4 processor can compute $555$ signatures and encrypt $448$ messages per second. In our experiments, we stored the precomputation tables on the SRAM.

\begin{table*}[t!]
	\centering
	\small
	\caption{Comparison of Public Key Primitives on 32-bit ARM Processor} \label{tab:ARM:PublicKey}
	\resizebox{\textwidth}{!}{
		\begin{threeparttable}
			\begin{tabular}{| c || c | c |  c | c | c | c |}
				\hline
				\textbf{Protocol} & \textbf{CPU Cycles} & \specialcell[]{\textbf{CPU Time} \\ \textbf{($ms$)}} & \specialcell[]{\textbf{Memory\tnote{$ \dagger $}} \\ \textbf{($Byte$)}} & \specialcell[]{\textbf{Bandwidth} \\ \textbf{($Byte$)}} & \specialcell[]{\textbf{Energy} \\ \textbf{Consumption ($mJ$)}} & \specialcell[]{\textbf{Certificate} \\ \textbf{Overhead}} \\ \hline

				\multicolumn{7}{|c|}{\textbf{\em Standard Techniques}} \\ \hline 
				\rowcolor{lightgray}
				ECDH & $8,710,000$ & $ 51.84 $ & $ 32 $ & $ 32 $ &$  6.84 $ & Yes \\
				
				Ephemeral ECDH & $ 17,611,000 $ & $ 104.83 $ & $ 32 $ & $ 32 $ & $ 13.84 $  &  Yes \\
				\rowcolor{lightgray}
				ECDSA-Sign & $ 9,412,000 $ & $ 56.02 $ & $ 32 $ & $ 64 $ & $ 7.40 $  & Yes \\
				
				ECDSA-Verify & $ 8,169,000 $ & $ 48.62 $ & $ 32 $ & $ 64 $ & $ 6.42 $ & Yes \\
				\rowcolor{lightgray}
				ECIES-Encrypt & $ 17,625,000 $ & $ 104.91 $ & $ 32 $ & $ 32 $ + |c| + |MAC| & $ 13.85 $ & Yes \\
				
				ECIES-Decrypt & $ 8,818,000 $ &$  52.49  $& $ 32 $ & $ 32 $ + |c| + |MAC| &$  6.93 $ & Yes \\ \hline
				
				\multicolumn{7}{|c|}{\textbf{\em \Dronecrypt}} \\ \hline 
				\rowcolor{lightgray}
				\texttt{AQ} & $ 556,000 $ &$  3.33 $ & $ 32 $ & $ 32 $ &$  0.44 $ &  No \\
				
				\texttt{BPV-AQ-Hang} & $ 764,000 $ & $ 4.55 $ & $ 16416 $ & $ 32 $ & $ 0.60 $ &  No \\
				\rowcolor{lightgray}
				$\bpvschsig$ & $ 302,000 $ & $ 1.80 $ & $ 16416 $ & $ 64 $ &$  0.24 $ & No \\
				
				$\bpvschver$ & $ 695,000  $ & $ 4.14  $& $ 32 $ & $ 64 $ & $ 0.55 $ & No\\
				\rowcolor{lightgray}
				$\bpviessig$ & $ 374,000 $ & $ 2.23  $& $ 24608 $ & $ 32 $ + |c| + |MAC| & $ 0.29 $  & No \\
				
				$\bpviesver$ &$  570,000 $ & $ 3.39 $ & $ 32 $ & $ 32 $ + |c| + |MAC| & $ 0.45 $ & No \\ \hline

			\end{tabular}
			\begin{tablenotes}[flushleft] \scriptsize{
					$\dagger $ Memory denotes private key size for sign/encrypt schemes as signer/sender stores it, memory denotes public key size for verify/decrypt schemes as verifier/receiver stores it.
				}
			\end{tablenotes}
		\end{threeparttable}
	}
\end{table*}

\begin{table}[t!]
	\centering
	\small
	\caption{Comparison of Symmetric Ciphers on 32-bit ARM} \label{tab:ARM:Symmetric}
	\vspace{-1mm}
	\resizebox{0.49\textwidth}{!}{
	\begin{threeparttable}
		\begin{tabular}{| c || c| c | c | c |}
			\hline
			\textbf{Protocol} & \textbf{($MB/s$)} & \specialcell[]{\textbf{CPU} \\ \textbf{Cycles\tnote{$ \dagger $}}} & \specialcell[]{\textbf{CPU} \\ \textbf{Time ($\mu s$)}} &  \specialcell[]{\textbf{Energy} \\ \textbf{Cons. ($\mu J$)}} \\ \hline

			\multicolumn{5}{|c|}{\textbf{\em Standard Techniques}} \\ \hline 
			
			AES & $ 0.926 $ & $ 5,540 $ & $ 32.96 $ &$  4.35 $\\
			
			\rowcolor{lightgray}
			AES-GCM & $ 0.377 $ &$  13,600 $ & $ 80.95 $ &$  10.68 $ \\
			
			HMAC\tnote{$ \ddagger $} &$  3.338  $& $ 1,540 $ & $ 9.14 $ &$  1.21 $ \\ \hline 
			
			\multicolumn{5}{|c|}{\textbf{\em \Dronecrypt}} \\ \hline 
			
			CHACHA20 & $ 3.554 $ &$  1,440 $&$  8.59 $ & $ 1.13 $ \\
			\rowcolor{lightgray}
			CHACHA-POLY &$  2.619 $& $ 1,960$ & $11.65$ & $1.54 $ \\
			
			POLY1305 & $ 13.709 $ & $ 370 $ &$  2.23 $ &  $ 0.29  $\\  \hline

		\end{tabular}
		\begin{tablenotes}[flushleft] \scriptsize{
				$\dagger $ CPU cycles presented here are for a $32$-byte message. \\
				$ \ddagger $ SHA256 is used as the standard hash function for HMAC.
			}
		\end{tablenotes}
	\end{threeparttable}
	}
\end{table}

\subsection{Performance on Commodity Hardware}

Considering the recipient side (e.g., command centers, ZSPs, and servers) that are involved in IoD networks are equipped with high-end processors, we also measured the costs of \Dronecrypt~on a commodity hardware. In our experiments, we used an Intel i7 Skylake $2.6$ $GHz$ CPU running with Linux 16.04. Our results showed that these processors can support up to a very large number of message throughputs so that they can communicate with a large drone fleet easily. More specifically, signing a message (with $\bpvschsig$) only takes around $10$ $ \mu s $ and encrypting a message (with $\bpviessig$) takes around $15$ $ \mu s $ . Verification and decryption algorithms in \Dronecrypt~only take $40-45$ $ \mu s $ . Therefore, \Dronecrypt~supports the secure communication of thousands of messages, and thereby permits a better quality-of-service and response time, wherein recipients (e.g., ZSPs) may handle many drones simultaneously. We open-sourced our codes on the 64-bit processor as well, to complete our framework.

\section{Conclusion}\label{sec:Conclusion}

In this paper, we propose \Dronecrypt, an efficient public-key based cryptographic framework that is tailored for the stringent energy and network requirements of IoD. \Dronecrypt~eliminates (i) the conventional PKI overhead with self-certified keys and (ii) EC scalar multiplications with BPV and DBPV techniques, and adopts FourQ curve that offers very fast curve operations (e.g., EC addition). We prove that our improved schemes in \Dronecrypt~are secure, and provide their  implementation on two common drone processors, namely 8-bit AVR and 32-bit ARM. Our experiments showed that \Dronecrypt~offers up to $13$$\times$ and $48$$\times$ less energy consumption compared to standard techniques, on 8-bit AVR and 32-bit ARM, respectively.

\bibliographystyle{IEEEtran}
\bibliography{crypto-etc,../../../Cryptoetc/crypto-etc}

\begin{IEEEbiography}[
	{
		\includegraphics[width=1in,height=1.25in,clip,keepaspectratio]{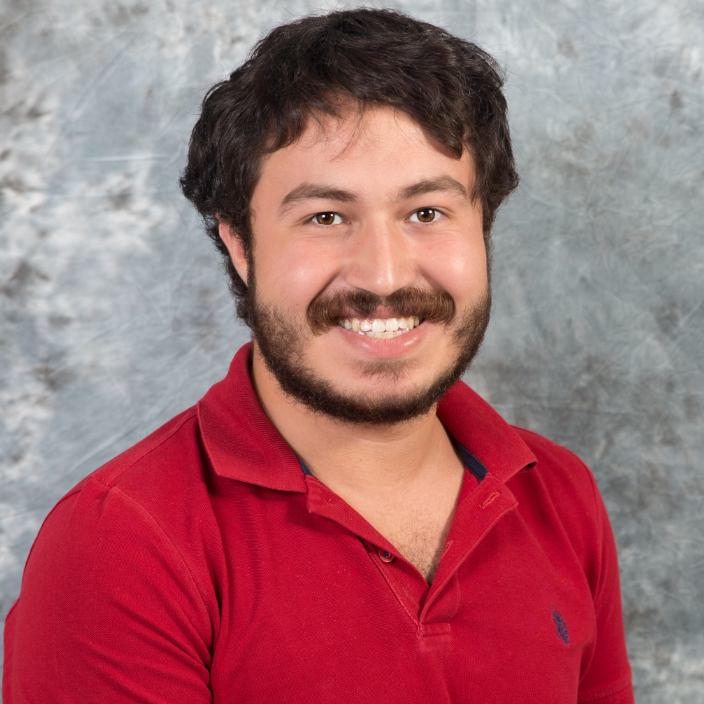}
	}
	]{Muslum Ozgur Ozmen}
	received the bachelor`s degree in electrical and electronics engineering from the Bilkent University, Turkey and the M.S. degree in computer science from Oregon State University. He is currently pursuing a PhD degree in computer science with the Department of Computer Science and Engineering, University of South Florida. His research interests include lightweight cryptography for IoT systems (drones and medical devices), digital signatures, privacy enhancing technologies (dynamic symmetric and public key based searchable encryption) and post-quantum cryptography.
\end{IEEEbiography}
\vskip 0pt plus -1fil

\begin{IEEEbiography}[
	{
		\includegraphics[width=1in,height=1.25in,clip,keepaspectratio]{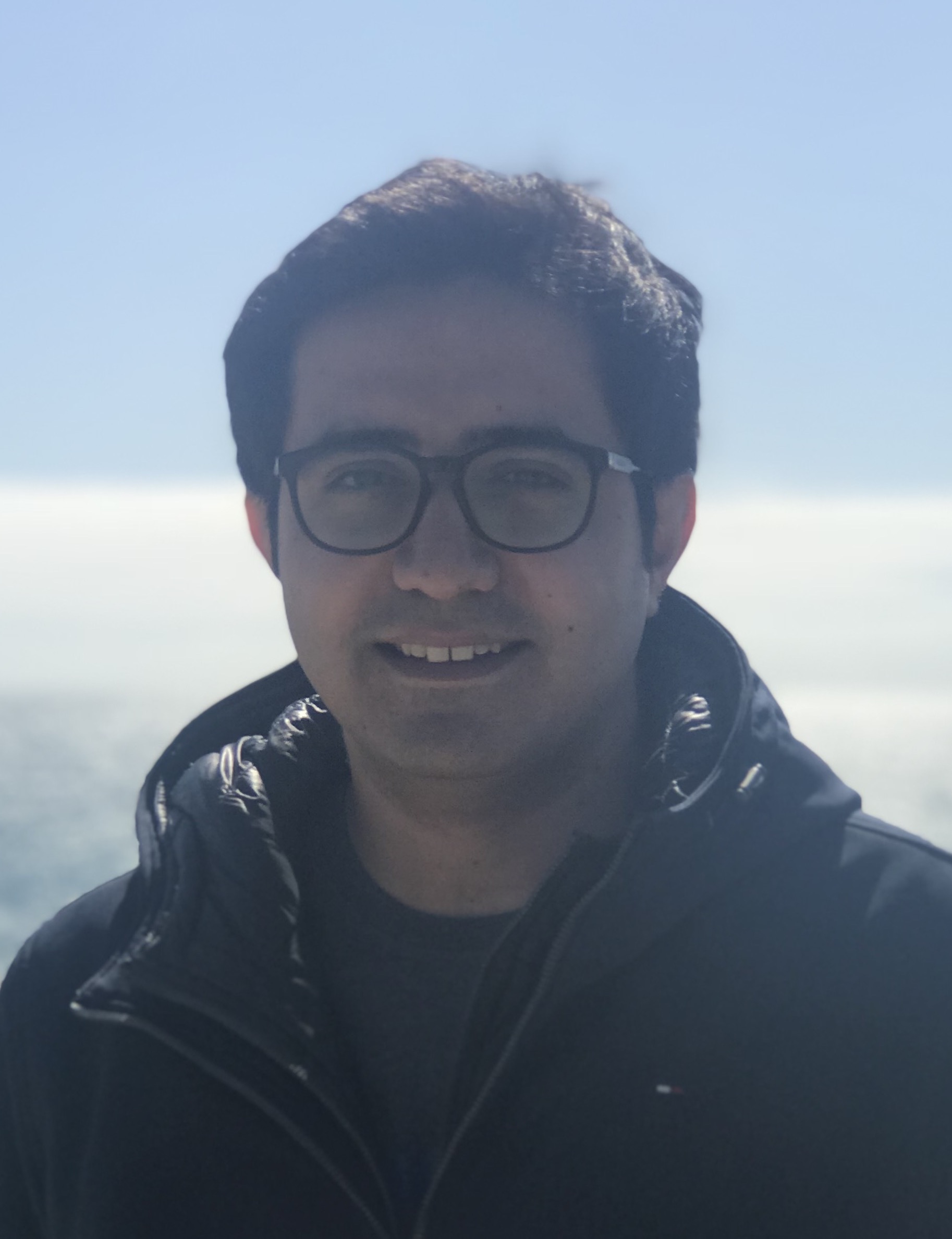}
	}
	] {Rouzbeh Behnia}
	received an M.S. degree from Multimedia University, Malaysia in 2013. He was a recipient of an FRGS grant during his time as a lecturer in Multimedia University and was awarded a Researcher Award in 2016. In 2016, he received an Outstanding Graduate Scholar award from the School of  Electrical Engineering and Computer Science, Oregon State University. He is currently pursuing a Ph.D. in Computer Science at the University of South Florida.  His research interests include post-quantum cryptography, privacy enhancing technologies and efficient authentication schemes.
\end{IEEEbiography}
\vskip 0pt plus -1fil

\begin{IEEEbiography}[
	{
		\includegraphics[width=1in,height=1.25in,clip,keepaspectratio]{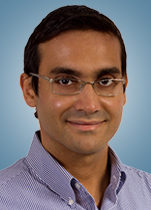}
	}
	]
	{Attila Altay Yavuz}
	(M `11) is an Assistant Professor in the Department of Computer Science and Engineering, University of South Florida (August 2018). He was an Assistant Professor in the School of Electrical Engineering and Computer Science, Oregon State University (2014-2018). He was a member of the security and privacy research group at the Robert Bosch Research and Technology Center North America (2011-2014). He received his PhD degree in Computer Science from North Carolina State University in August 2011. He received his MS degree in Computer Science from Bogazici University (2006) in Istanbul, Turkey. He is broadly interested in design, analysis and application of cryptographic tools and protocols to enhance the security of computer networks and systems. Attila Altay Yavuz is a recipient of NSF CAREER Award (2017). His research on privacy enhancing technologies (searchable encryption) and intra-vehicular network security are in the process of technology transfer with potential world-wide deployments. He has authored more than 40 research articles in top conferences and journals along with several patents. He is a member of IEEE and ACM.
\end{IEEEbiography}

\end{document}